\newif\iflong 
\longtrue
\iflong
\documentclass[a4paper,USenglish,cleveref,final,nolineno%
]{socg-lipics-v2021} 

\iffalse 
\newcommand{\Therese}[1]{\textcolor{brown}{TB: #1}}
\newcommand{\Guenter}[1]{\textcolor{blue}{GR: #1}}
\else
\newcommand{\Therese}[1]{}
\newcommand{\Guenter}[1]{}
\fi

\hideLIPIcs
 \makeatletter
\let\@oddfoot\@empty 
\makeatother
\else
\documentclass[a4paper,USenglish,cleveref]{socg-lipics-v2021} 
\fi

\usepackage{amssymb}
\usepackage{amsmath}
\usepackage{amsthm}
\usepackage{hyperref}
\usepackage[textsize=scriptsize]{todonotes}

\usepackage{url}
\usepackage{color}

\usepackage{graphicx}

\usepackage{xspace}

\usepackage[ruled,vlined,linesnumbered,titlenumbered]{algorithm2e}
\DontPrintSemicolon
\SetKwInput{Input}{Input}
\SetKwInput{Output}{Result}
\SetKw{KwBreak}{break}
\SetKw{Break}{break}
\SetKw{Continue}{continue}
\SetKwRepeat{Do}{do}{while}

\newcommand{\LL}{\ensuremath{L}\xspace} 

\newcommand{\B}{{\ensuremath{B}}\xspace}
\newcommand{\state}{\textit{state}}

\newcommand{\available}{\ensuremath{\text{\tt free}}\xspace}

\newcommand{\blocked}{\text{\tt blocked}\xspace}
\newcommand{\clique}{\textit{clique}}
\newcommand{\calT}{{\ensuremath{\mathcal{T}}}\xspace}
\newcommand{\D}{\ensuremath{D}\xspace}  

\newcommand{\I}{\ensuremath{I}\xspace}  
\newcommand{\depth}{\ensuremath{\mathrm{depth}}\xspace}  
\newcommand{\search}{\ensuremath{\textit{search}}\xspace}  
\newcommand{\highest}{\ensuremath{z_{\text{high}}}\xspace}  
\newcommand{\dom}{\ensuremath{z_{\text{dom}}}\xspace}  
\newcommand{\undom}{\ensuremath{z_{\text{undom}}}\xspace}  
\newcommand{\h}{\ensuremath{h}\xspace}  
\renewcommand{\d}{\ensuremath{d}\xspace}  
\renewcommand{\u}{\ensuremath{u}\xspace}  

\newcommand{\NIL}{\text{\tt nil}\xspace}  
\newcommand{\node}{\ensuremath{a}\xspace}  
\newcommand{\parent}{\text{\rm parent}\xspace}
\newcommand{\Find}{\text {\sc Find}\xspace}
\newcommand{\Union}{\textup{\textsc{Union}}\xspace}
\newcommand{\M}{\ensuremath{M}\xspace}

\newcommand{\therese}[1]{\Therese{#1}}  

\title{Solving Problems in RDV Graphs Faster and Easier}

\author{Therese Biedl}{David R.~Cheriton School of Computer Science,
University of Waterloo,
Canada}{biedl@uwaterloo.ca}{https://orcid.org/0000-0002-9003-3783}{Supported
by NSERC, FRN RGPIN-2020-03958.   Research done while the first author visited Freie Universität Berlin 
during her sabbatical in the winter term 2025/26.}
\author{Günter Rote}{Institut für Informatik,
Freie Universität Berlin
}{rote@inf.fu-berlin.de}{https://orcid.org/0000-0002-0351-5945}{}

\Copyright{Therese Biedl and Günter Rote}



\let\paragraph=\subparagraph
\begin{document}
\maketitle

\begin{abstract}
In interval graphs and many generalizations thereof, problems such as {\sc Matching}, {\sc IndependentSet} and {\sc DominatingSet} can be solved with simple greedy algorithms in linear (i.e., $O(m+n)$) time.
However, it is sometimes possible to {improve} on this runtime
if a suitable representation as intersection graph is given; for example
it is possible to find a minimum dominating set in an interval graph
in $O(n)$ time.
We study RDV graphs, a generalization of interval graphs, namely, the intersection graphs of  directed paths in a rooted tree. It was known that the above three problems can be solved in $O(n\log n)$ time if a suitably small representation is given. 
In this paper,
we reduce the runtime to $O(n)$.
Furthermore, while the previous results
used non-trivial operations 
such as range-minimum queries and
ray-shooting queries, we need nothing more
complicated than lists and (for {\sc Matching}) a \Union-\Find data structure.
\end{abstract}

\section{Introduction}

We study efficient algorithms for graphs that are
\emph{RDV graphs}, i.e., {\bf v}ertex-intersection graphs of
{\bf d}irected paths in a {\bf r}ooted tree.   (Formal definitions are
given in Section~\ref{sec:background}.)   This graph class
 contains interval graphs and is a subclass of the chordal graphs;
so any problem that can be solved efficiently in chordal graphs
(such as {\sc IndependentSet}, which is NP-hard in general but
can be solved with a linear-time greedy algorithm in chordal graphs \cite{Gavril72}) can
also be solved efficiently in RDV graphs.   In fact,
RDV graphs are a subclass of so-called strongly chordal graphs.
In strongly chordal graphs, problems such as {\sc DominatingSet}
and {\sc Matching} can be solved in linear time with a greedy algorithm
\cite{DahlhausKarpinski,Farber84},
while no linear-time algorithms are known for chordal graphs (and in fact
{\sc DominatingSet} is NP-hard in chordal graphs \cite{BoothJ82}).
\iflong\else\looseness-1 \fi

In the above algorithms, ``linear'' refers to the input size, i.e., the runtime is $O(m+n)$ for a graph with $n$ vertices and $m$ edges.   In contrast to this, later papers considered 
algorithms whose runtime depends only on $n$, assuming that a suitably small representation of the graph as intersection graph of objects is given.   Of particular interest to the current work is  a recent string of papers where the first author and Gokhale  
argued that the runtime of the above greedy algorithms can be reduced to $O(n\log n)$ in RDV graphs 
\cite{BG-CCCG24,BG-CGT26,BG-SOFSEM26,PrashantThesis}.%
\footnote{{\sc Matching} is addressed in
\iflong
the conference paper
  \cite{BG-CCCG24}, and in its journal version
  \cite{BG-CGT26} with a simplified approach.
  \else
  \cite{BG-CCCG24} and
  \cite{BG-CGT26}.
  \fi
  {\sc DominatingSet} is addressed in \cite{BG-SOFSEM26}, while {\sc IndependentSet} is addressed only in Gokhale's Master's thesis \cite{PrashantThesis}, which summarizes all these results.}

Specifically, the runtime of these  algorithms is $O(|\calT|+n\log n)$,
where $\calT$ is the tree used
for the representation of~$G$.   There always exists a
representation with $|\calT|\in O(n)$, and so the runtime is
$O(n\log n)$ if the given representation is small enough
(an assumption that we will make throughout). 

The crucial idea for these algorithms with runtime depending only on $n$ and not on $m$
is to reframe a given tree representation in geometric terms,
namely as an intersection graph
of T-shapes. Biedl and Gokhale define for every vertex of $G$ a horizontal
and a vertical segment whose union forms a T such that 
two vertices are adjacent if and only if their Ts intersect.
One can then express the operations needed by the greedy algorithm 
as queries among these segments, and the authors showed
that it suffices to use
range-maximum queries and orthogonal ray-shooting queries in a 
(dynamically changing) subset of the segments. 
Prior results allow
these queries to be implemented with runtime $O(\log n)$:
McCreight's priority search trees support range-maximum queries
\cite{McCreight85}, and Giyora and Kaplan gave 
data structures for orthogonal ray-shooting queries \cite{GiyoraK09}. 
However, neither of these data structures is
straightforward to implement, 
especially in a dynamic setting.

In this paper, we improve on these algorithms
for RDV graphs in three ways.   First, we show that the detour
into the T-intersection representation is not necessary: everything
can be done directly with a bottom-up traversal of the tree $\calT$
of the tree representation.

Second,
the runtimes are better: we achieve runtime $O(n)$ for all three problems,
compared to the previous bound of $O(n\log n)$
\cite{BG-SOFSEM26,BG-CGT26,PrashantThesis}.
For 
{\sc Matching}, our runtime even
beats the best previously known result for \emph{interval graphs},
which was $O(n\log\log n)$ 
(Moitra and Johnson~\cite{MoitraJ89} with an improvement
pointed out by  Rhee and
Liang \cite[Section~5]{RheeLiang}).

Finally, and this is what
we view as our most important contribution, 
 our algorithms are very simple.
There is no need for
complicated data structures; all the required information can be
obtained by storing at each node of $\calT$ a few integers or lists
or vertex references.   (For {\sc Matching} we also use a 
\Union-\Find data structure, but this is easy to implement 
if we relax the runtime bound
from $O(n)$ to $O(n\alpha(n))$).

\paragraph{Organization of the paper.}
Section~\ref{sec:background} introduces RDV representations and our
notation, and Section~\ref{sec:structure} describes the bottom-up
traversal of $\calT$ on which our algorithms are based.
Sections~\ref{sec:mis}, \ref{sec:mds3}, and~\ref{sec:match} give our
algorithms for {\sc IndependentSet}, {\sc DominatingSet}, and
{\sc Matching}, respectively, each with a small example, and
Section~\ref{sec:outlook} mentions open problems.
Appendix~\ref{sec:streamlined} contains streamlined pseudocode for all
three algorithms, and Appendix~\ref{sec:equiv} proves that the two
greedy algorithms for {\sc Matching} that we build on produce the same
matching, and that this matching is maximum.
%

\section{Background}
\label{sec:background}

We assume familiarity with graphs, and refer e.g.\ to Diestel's and Golumbic's books for basic terminology \cite{Die12,Gol80}.   
In this paper we study graphs $G$ that have a \emph{representation as a vertex-intersection graph of paths in a tree}. 
This consists of a \emph{host tree} $\calT$ and, for each vertex $z$ of $G$, a path $P(z)$ in $\calT$ such that $zy$ is an edge of $G$ if and only 
if $P(z)$ and $P(y)$ share at least one node of~$\calT$.
By convention, we use the term `node' and early letters $a,b,c,\dots$
of the alphabet
for elements of $V(\calT)$, to distinguish them from the vertices of $G$ for which we use $z,y,x,\dots$\iflong. \else\looseness-1 \fi

The class of graphs with such a representation is known as the \emph{path graphs}.   We study here a graph
class that is a proper subclass of the path graphs; it was named by Monma and Wei \cite{MonmaWei},
though studied even earlier e.g.~by Gavril \cite{Gavril75} and Booth
and Johnson \cite{BoothJ82}:\footnote
{They called RDV graphs \emph{directed path graphs}, but this later came to
mean a different graph class
, where the edges of $\calT$
can be oriented independently without regard to a common root.
}

\begin{definition} A \emph{rooted directed path graph} \textup(or \emph{RDV graph} \textup{\cite{MonmaWei})}
is a graph that has an \emph{RDV representation}, i.e., a
representation as vertex-intersection graph of
directed paths in a rooted tree $\calT$ in which
all edges are directed towards the root.
\end{definition}

\usetikzlibrary{backgrounds}
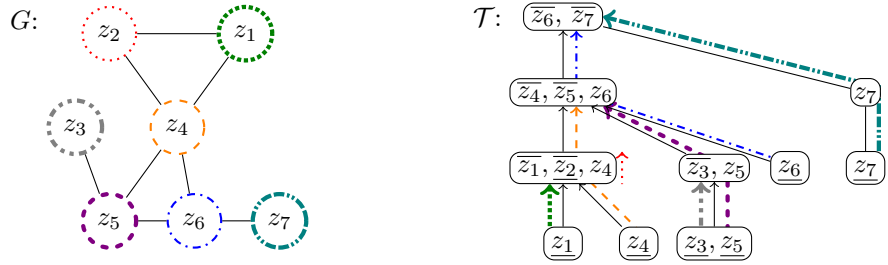
\begin{figure}[ht]
\centering
\begin{tikzpicture}[
        host/.style = {draw,rectangle,rounded corners,fill=white,inner sep=1pt},
        graph/.style = {draw,circle,fill=white},
        z1/.style = {draw=green!50!black,ultra thick,densely dotted},
        z2/.style = {draw=red,thick,dotted},
        z3/.style = {draw=gray,ultra thick,dash dot dot},
        z4/.style = {draw=orange,thick,dashed},
        z5/.style = {draw=blue!50!red,ultra thick,loosely dotted,cap=round},
        z6/.style = {draw=blue,thick,dash dot},
        z7/.style = {draw=green!50!blue,ultra thick,dash pattern={on 4pt off 1pt on 4pt off 1pt on 1pt off 1pt on 1pt off 1pt}},
        level distance = 1cm,
        sibling distance = 1cm,
]
\begin{scope}[scale=0.45]
\node [draw=none] at (-2.5,7.9) (G) {$G$:};
    \node[graph,z3] (z3) at (-1, 4.75) {\( z_3 \)};
    \node[graph,z1] (z1) at (4, 7.5) {\( z_1 \)};
    \node[graph,z2] (z2) at (0, 7.5) {\( z_2 \)};
    \node[graph,z4] (z4) at (2, 4.75) {\( z_4 \)};
    \node[graph,z5] (z5) at (0, 2) {\( z_5 \)};
    \node[graph,z6] (z6) at (2.5, 2) {\( z_6 \)};
    \node[graph,z7] (z7) at (5, 2) {\( z_7 \)};
    \draw (z3) -- (z5);
    \draw (z5) -- (z6);
    \draw (z6) -- (z7);
    \draw (z6) -- (z4);
    \draw (z4) -- (z1);
    \draw (z5) -- (z4);
    \draw (z2) -- (z4);
    \draw (z2) -- (z1);
\end{scope}
\begin{scope}[scale=1,yshift=4.6cm,xshift=5cm]
\node [draw=none] at (0,-1) (T) {$\calT$:};

\node [host] (r) at (1,-1) {$\,\overline{z_6}$, $\overline{z_7}\,$}
        child {edge from parent[draw=none,<-]}
        child {edge from parent[draw=none]}
        child {edge from parent[draw=none]}
        child {edge from parent[draw=none]}
        child { node [host] (c2) {$\,\overline{z_4},\overline{z_5},{z_6}$} [<-]
                child {edge from parent[draw=none]}
                child {edge from parent[draw=none]}
                child {edge from parent[draw=none]}
                child { node [host] (g3) {$\,\overline{z_1},\underline{\overline{z_2}},z_4$}
                        child {edge from parent[draw=none]}
                        child { node [host] (b3) {$\,\vphantom A\underline{z_1}\,$}}
                        child { node [host] (l4) {$\,\vphantom A\underline{z_4}$\,}}
                }
                child {edge from parent[draw=none]}
                child { node [host] (g2) {$\,\overline{z_3},z_5$}
                        child { node [host] (l2) {\vphantom A$\,\underline{z_3},\underline{z_5}\,$}}
                }
                child { node [host] (l5) {$\,\underline{z_6}\,\vphantom A$} }
        }
        child {edge from parent[draw=none]}
        child {edge from parent[draw=none]}
        child {edge from parent[draw=none]}
        child { node [host] (c1) {$z_7\vphantom A$}
                child { node [host] (l1) {$\,\underline{z_7}\,\vphantom A$} }
        }
;
\end{scope}

\begin{scope}[on background layer]
\draw [z3,<-] ([xshift=-5pt]g2.south)-- ([xshift=-5pt]l2.center); 
\draw [z1,<-] ([xshift=-5pt]g3.south)-- ([xshift=-5pt]b3.center); 
\draw [z2,<-] ([xshift=2pt]g3.north east) -- ([xshift=2pt]g3.south east); 
\draw [z4,<-] ([xshift=5pt]c2.south) -- ([xshift=5pt]g3.center)-- ([xshift=5pt]l4.center); 
\draw [z5,<-] ([xshift=-5pt]c2.south east) -- ([xshift=5pt]g2.center) -- ([xshift=5pt]l2.center); 
\draw [z6,<-] ([xshift=5pt]r.south) -- ([xshift=5pt]c2.center) -- ([xshift=5pt]l5.center); 
\draw [z7,<-] ([yshift=-0pt]r.east) -- ([xshift=5pt,yshift=2pt]c1.center) -- ([xshift=5pt]l1.center); 
\end{scope}
\end{tikzpicture}
\caption{
An RDV graph $G$ with an RDV representation~$\calT$.
The bags are written in the nodes of~$\calT$. (Picture 
taken from \cite{BG-SOFSEM26}.)
}
\label{fig:RDV_example}
\end{figure}

\begin{figure}[ht]
\centering
\begin{tikzpicture}[x=1.75cm,y=0.95cm,
        host/.style = {draw,rectangle,rounded corners,fill=white,inner sep=2pt},
        tree/.style = {->,>=latex,thin}]
\foreach \d in {0,...,6} \node[font=\scriptsize,gray] at (-0.2,-\d) {\d};
\node[font=\scriptsize,gray] at (-0.2,0.7) {depth};
\node[host] (r)  at (3,0)    {$\overline{z_6}$};
\node[host] (s)  at (3,-1)   {$z_6,\overline{z_7}$};
\node[host] (l7) at (1.6,-2) {$\underline{z_7}$};
\node[host] (p)  at (4,-2)   {$\overline{z_4},z_6$};
\node[host] (q)  at (4,-3)   {$z_4,\overline{z_5},z_6$};
\node[host] (x)  at (2.5,-4) {$\overline{z_1},z_4$};
\node[host] (y)  at (2.5,-5) {$z_1,\underline{\overline{z_2}},z_4$};
\node[host] (b1) at (1.9,-6) {$\underline{z_1}$};
\node[host] (b4) at (3.1,-6) {$\underline{z_4}$};
\node[host] (g)  at (4.4,-4) {$\overline{z_3},z_5$};
\node[host] (b3) at (4.4,-5) {$\underline{z_3},z_5$};
\node[host] (b5) at (4.4,-6) {$\underline{z_5}$};
\node[host] (b6) at (5.5,-4) {$\underline{z_6}$};
\foreach \c/\p in {s/r,l7/s,p/s,q/p,x/q,g/q,b6/q,y/x,b1/y,b4/y,b3/g,b5/b3}
  \draw[tree] (\c) -- (\p);
\foreach \v/\k in {l7/1,b1/2,b4/3,y/4,x/5,b5/6,b3/7,g/8,b6/9,p/11,r/13}
  \node[font=\scriptsize,text=blue!60!black,anchor=north west,inner sep=1pt] at (\v.south east) {\k};
\foreach \v/\k in {q/10,s/12}
  \node[font=\scriptsize,text=blue!60!black,anchor=west,inner sep=1pt] at ([yshift=-1pt]\v.east) {\k};
\end{tikzpicture}
\caption{The running example: an RDV representation of the graph $G$ of
  Figure~\ref{fig:RDV_example} that satisfies our assumptions.
  Each node shows its bag; top nodes are marked by overlining and bottom
  nodes by underlining the corresponding vertex.
  The top nodes are distinct, the bottom nodes are distinct, and the
  root is a top node.  The depths of the nodes are shown on the left.
  Since every node is a top node or a bottom node, we refer to the nodes
  as $t(z_i)$ and $b(z_i)$.
  The small numbers give the order in which the nodes are processed in
  the examples.}
\label{fig:running}
\end{figure}
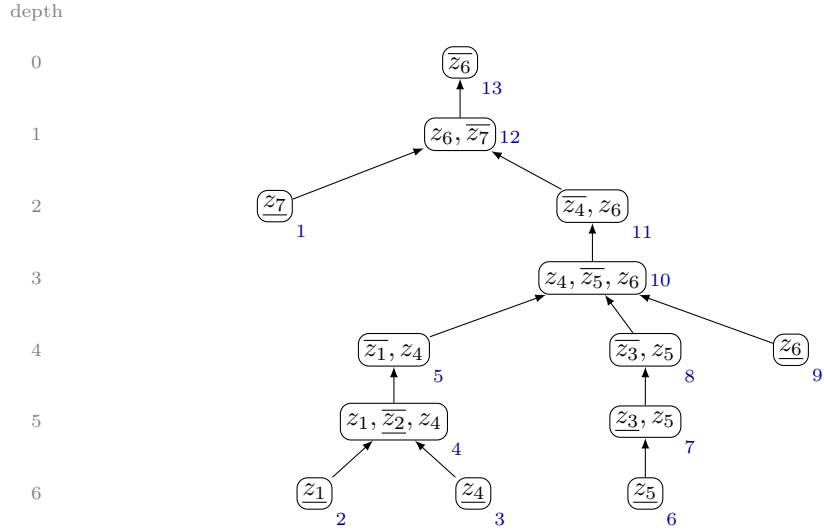

Figure~\ref{fig:RDV_example} illustrates an RDV graph and its RDV representation.    RDV graphs can be recognized in polynomial time \cite{Gavril75}, and along the way, an RDV representation of size $O(n)$ is produced.
We use $t(z)$ and $b(z)$ for the \emph{top node} and \emph{bottom node} of the path of a vertex $z$, respectively.    
We may without loss of generality assume that $t(z)\neq t(y)$ for any two vertices $z,y$; if this
is violated (say node $\node$ is both $t(z)$ and $t(y)$) then we can subdivide the edge from $\node$ to
its parent (or add a parent if $\node$ is the root) and extend the path of $z$ upward into the new node
without affecting the represented graph.   Similarly we can require that $b(z)\neq b(y)$ for any two
vertices $z,y$, by adding a new child to $b(z)$ if required.
Finally we may assume that the root is a top node, for if it is not then no path uses it, and so we can simply contract one of its children with the root without affecting the represented graph.
These assumptions make the statements of the algorithms more
convenient.
Figure~\ref{fig:running} shows an RDV representation of the graph of
Figure~\ref{fig:RDV_example} that satisfies all of them; we will use it as
a running example.\footnote{
\iflong We could also stipulate that no node
  is both 
  a top node and a bottom node, but we
  do not need this. 
  \fi 
  A degenerate path consisting of a single node is
\iflong also \fi permitted.}

We assume that every vertex $z$ can look up
nodes $t(z)$ and $b(z)$ in constant time.
We also prepare 
$\calT$ such that each node $\node$ 
can do the following in constant time:
\begin{enumerate}
\item Look up $\depth(\node)$, which is the distance
of $\node$ from the root of~$\calT$.
\item Test whether $\node$ is the top node of some $z\in V(G)$, and if so, look up this (unique) vertex~$z$.
\item Test whether $\node$ is the bottom node of some $z\in V(G)$, and if so, look up this (unique) vertex~$z$.
\end{enumerate}

For any $\node\in V(\calT)$, we write $\beta(\node)$ for the set of vertices whose
paths contain 
$\node$, i.e., $\beta(\node)=\{z \in V(G): \node\in P(z)\}$.
We call $\beta(\node)$ the \emph{bag} of $\node$, a name
inspired by the fact that $\calT$ and the assignment of paths is a special case of a so-called tree
decomposition (see e.g.~\cite{Die12}) where the corresponding concept
is called a bag.
We will need a simple observation. 

\begin{observation}
\label{obs:inside}
Consider a vertex $z$ and its top node $\node=t(z)$.
Every neighbor $w$ of $z$ either is in $\beta(\node)$, or 
$t(w)$ is a
strict descendant of~$\node$.
\end{observation}
\begin{proof}
Assume that $w\not\in \beta(\node)$, so $P(w)$ does not contain~$\node$.
But $P(w)$ must share a node with the path $P(z)$, say node $b\neq a$.   This must be a descendant of $a$ since
$a$ is the top node of $P(z)$,    and it is a \emph{strict} descendant since $b\neq a$. So $P(w)$ contains 
strict descendants of $a$, but does not contain~$a$;
thus $t(w)$ must be a strict descendant of~$\node$.
\end{proof}

We will
use the common notation  $N[z]$ for the closed neighborhood of a
vertex $z
$,
consisting of $z$ together with its 
adjacent vertices.
For a set of vertices $W
$,
 $N[W]$ denotes the
 union of $N[z]$ over all $z\in W$.
We study the following three problems:
\begin{itemize}
\item {\sc IndependentSet}:  Given a graph $G$, find a maximum vertex set $\I$ that
	is \emph{independent}, i.e., no edge connects two vertices in~$I$.
\item {\sc DominatingSet}:  Given a graph $G$, find a minimum vertex set $\D$ 
  that is \emph{dominating}, i.e., 
  each vertex of $G$ that is not in $D$ has a neighbor in~$D$, or more formally,
   $N[\D]=V(G)$.
      \item {\sc Matching}:  Given a graph $G$, find a maximum edge set $\M$
        such that no two edges of $\M$ have a common endpoint.
\end{itemize}

We will not review the extensive history of these three problems, and what
is known in terms of algorithms and/or NP-hardness, except for graph classes
that are closely related to RDV graphs, where such a review will be given in
the relevant section below.

\section{The Common Structure of Our Algorithms}
\label{sec:structure}

Our algorithms for  {\sc IndependentSet} and {\sc DominatingSet}
 have a common basic structure: They perform a bottom-up traversal of the tree~$\calT$.
At each node, the same 4-step procedure is used.
Our algorithm for {\sc Matching} also processes the nodes of $\calT$
from the leaves to the root, but it deviates from this
scheme: it does not work with the bags, but with a \Union-\Find data
structure on the nodes of~$\calT$ (Section~\ref{sec:match}).
We illustrate this \iflong paradigm \fi with a skeleton that 
makes the bags $\beta(\node)$ available when needed.
\iflong \else \looseness-1 \fi

\paragraph{Processing node~$\node$.}  We assume that all
strict descendants of $\node$ have already been processed.
\begin{enumerate}
\item (Accumulate the data from the children.)  
  \begin{itemize}
  \item If $\node$ is a leaf, initialize an empty list $B(\node)$.
  \item  Otherwise, let $\B(\node)$ be the concatenation of lists $\B(c)$ over all children $c$ of~$\node$.  
    
  \end{itemize}
\item (Update if a path starts.)
If $\node$ is a bottom node, say $\node=b(z)$:   
  \begin{itemize}
	\item Add $z$ to $\B(\node)$.   
  \end{itemize}
\item (Make decisions.)\iflong\par\fi
  $\ldots$ [At this point, $B(\node)=\beta(\node)$.]
\item (Update if a path ends.)
If $\node$ is a top node, say $\node=t(z)$:   
  \begin{itemize}
	\item Remove $z$ from $\B(\node)$. 
  \end{itemize}
\end{enumerate}

\begin{lemma}\label{lem:beta}
For any node $\node$, the set $B(\node)$ is always a subset of $\beta(\node)$, and
  at the beginning of Step~3 when $\node$ is processed, $B(\node)=\beta(\node)$.
\end{lemma}
\begin{proof}
In fact, the manipulations of $B(\node)$ in Steps~1, 2, and~4 play
together to ensure that
the correct set $B(\node)$ is built inductively.
In Step~4, a vertex $z$ with $t(z)=\node$ is removed from $B(\node)$ in
preparation for building the set $B$ for the parent.
Thus in Step~1 the set at each child contains all those vertices whose paths contain the child and continue upward.    So all these paths continue into $\node$, hence the vertices are in $\beta(\node)$.  Clearly the vertex added in Step~2 (if any) is also in $\beta(\node)$. 
Vice versa, any vertex in $\beta(\node)$ corresponds to a path that contains~$\node$.   This path either had its bottom node at $\node$
(such vertices are added in Step~2), or it contains one of the children (and hence the vertex is added in Step~1).
\end{proof}

As it happens, the later algorithms will not actually 
{need}
the sets $B(\node)$, but they will compute subsets of $B(\node)$, and
we shall use $B(\node)$ for their correctness arguments. 

\paragraph{Leaf-to-root order.}
The only required condition on the order of processing is that a child
comes before its parent. Any order that is consistent with this can be
used in our algorithms, for example postorder or reverse level-order.

\paragraph{Optimality.}
Our three algorithms establish optimality in different ways.
For {\sc IndependentSet}, the algorithm computes, along with the
independent set, a clique cover of the same size, which certifies
optimality.
For {\sc DominatingSet}, we show that our algorithm computes the same
set as the greedy algorithm of Booth and Johnson~\cite{BoothJ82}, and
rely on their optimality proof.
For {\sc Matching}, we show that our algorithm computes the same
matching as a known greedy algorithm, whose optimality follows from
the general results of Dahlhaus and Karpinski~\cite{DahlhausKarpinski}
about strongly chordal graphs;
Appendix~\ref{sec:optimality} gives an independent self-contained
proof for our setting.

\section{Maximum Independent Set \texorpdfstring{$\I$}{I}}
\label{sec:mis}

It is well known that in a chordal graph, a maximum independent set can be found with a simple
linear-time greedy algorithm executed on a perfect elimination order (see Gavril's original paper \cite{Gavril72}
or Golumbic's monograph~\cite{Gol80} for details).  In an RDV graph, the
runtime can be reduced to $O(n\log n)$ by employing suitable data structures \cite[Thm. 5.2.2]{PrashantThesis}.
For interval graphs
where the interval endpoints are given in sorted order
(this corresponds to our assumption that the tree representation $\calT$ is given),
a maximum independent set $\I$ can be found
by a simple greedy algorithm in $O(n)$ time
\cite{Gupta-Lee-Leung82}.

In this section we give an algorithm that solves {\sc IndependentSet} in an RDV graph in $O(n)$ time.
Our algorithm will perform the same steps as the greedy algorithm, but rather
than showing this we will argue directly that it computes a maximum independent
set. In accordance with the paradigm in \Cref{sec:structure},
the algorithm performs a bottom-up traversal of $\calT$
and incrementally builds an independent set~$I$.

We maintain a
label $\state(z)
$ for every vertex $z$
with the following meaning:
\begin{description}
\item[(I1)] Every vertex $z$ in
  $N[I]$ has
  $\state(z)=\blocked$; all other vertices have $\state(z)=\available$.%
\end{description}
We also maintain 
a list $\LL(\node)$ for every 
node $\node$,
with the following invariant:
\begin{description}
\item[(I2)] $\LL(\node)=\{z\in \B(\node): \state(z)=\available\}$,
\end{description}
where $\B(\node)$ is the set introduced in \Cref{sec:structure}; recall
in particular that $\B(\node)=\beta(\node)$ during Step~3 when $\node$
is processed.
For reference, we include the manipulations of the sets $\B(\node)$
from the algorithm of \Cref{sec:structure}
as comments.

As an optimality certificate, we store for each $z\in I$ a
vertex set $\clique(z)$ with the following property (as the name suggests):
\begin{description}
    \item[(I3)] $\clique(z)$ forms a clique. 
\end{description}


\paragraph{Initialization.} Initialize an empty set $\I$ of vertices.
For every vertex $z \in V(G)$, initialize
$\state(z):= \available$.  
\paragraph{Processing node~$\node$.} We assume that all
strict descendants of $\node$  have already been processed.
\begin{enumerate}
\item (Accumulate the data from the children.)  
  \begin{itemize}
  \item  \textcolor{gray}{/\!/ Form $\B(\node)$ by concatenating the lists $\B(c)$ of the children $c$ of~$\node$.}
  \item  Form $\LL(\node)$ by concatenating the lists $\LL(c)$ of the children $c$ of~$\node$. 
  \end{itemize}
\item (Update if a path starts.)
If $\node$ is a bottom node, say $\node=b(z)$:
  \begin{itemize}
	\item  \textcolor{gray}{/\!/ Add $z$ to $\B(\node)$}\,
	\item Add $z$ to $\LL(\node)$ /\!/ $\state(z)=\available$
          at this point.

        \end{itemize}
\item (Add to solution if possible.)   If $\node$ is a top node, say $\node=t(z)$, and $\state(z)=\available$:
  \begin{itemize}
		\item Add $z$ to~$\I$.
\item /\!/ At this point, $\LL(\node)=\{z\in \beta(\node): \state(z)=\available\}$.
		\item For every vertex $w$ in $\LL(\node)$, set
                  $\state(w):=\blocked$. /\!/ This includes the vertex~$z$.
        \item                   Set $\clique(z):=\LL(\node)$, and
                  clear $\LL(\node)$ to be empty. 
  \end{itemize}
\item (Update if a path ends.)  If $\node$ is a top node, say $\node=t(z)$:
  \begin{itemize}
  \item  \textcolor{gray}{/\!/ Remove $z$ from $\B(\node)$.}\,
  /\!/ No update is needed for $\LL(\node)$.
  \end{itemize}
\end{enumerate}
After processing 
the root node, we return $\I$ as the maximum independent set.

\paragraph{Correctness.}
Invariant (I2). In Steps~1 and~2, the set $L(\node)$ is constructed and updated in accordance with  the
defining condition of $L(\node)$
that the state should be \available.
In fact,
in Step~2, $\state(z)$ must be \available, because the state
has not changed since initialization.
Hence the addition of $z$ to $L(\node)$ is correct.
In Step~4, it would be necessary 
to remove $z$ from $L(\node)$ in case
$\state(z)=\available$.
However, in this case,
 $L(\node)$ has already been emptied in Step~3.
Consequently, Step~4 is not necessary.

We make an easy but important observation.
\begin{description}
\item[(I4)]
All vertices $z$ for which
  $t(z)$ has been processed completely
  have $\state(z)=\blocked$.
\end{description}
Proof: When processing $t(z)$, the state of $z$ is either \available or \blocked.  If $z$ is \available, Step~3 adds it to $I$, which makes it \blocked; so either way $\state(z)=\blocked$ afterwards.

We now prove
the crucial property that the \state\ labels are maintained
correctly (Invariant~(I1)):
We mark vertices as \blocked if and only if they were in $L(\node)$
during Step~3  and we added the vertex $z\in L(\node)$ to~$I$.  Since
$L(\node)$ forms a clique, all these vertices are now in
$N[z]\subseteq N[I]$, and
therefore their classification as \blocked is correct.

We still have to show that \emph{every} vertex $w\in N[z]$ is \blocked
after Step~3 adds 
$z$ to~$I$:
\begin{itemize}
\item 
If $w\not\in \beta(\node)$, then by
Observation~\ref{obs:inside} its top node $t(w)$
is a strict descendant of $\node$ and has been processed before; by
(I4) this means that $w$ is already \blocked.
\item 
If $w\in \beta(\node)$ and is still \available when
processing $\node=t(z)$, then by (I2)  $w$ is in $L(\node)$, and we
mark it as \blocked.
\end{itemize}

We now turn to the sets
$\clique(z)$ (Invariant (I3)):  We set $\clique(z)=L(t(z))$ when processing
$\node=t(z)$. This is a clique since $L(\node) \subseteq B(\node)
\subseteq \beta(\node)$, and $\beta(\node)$ forms a clique. 

We make an additional observation, which will be important for the
optimality proof:
\begin{description}
\item[(I5)]
Every vertex $w$ with $\state(w)=\blocked$ is
      in $\clique(z)$ for a unique $z\in I$.
\end{description}
Proof: 
Inspection of Step~3 shows that a vertex $y$ is added
to some $\clique(z)$ exactly when its state changes from \available to
\blocked.
Since states never change back to \available, $y$ will never be added to another set $\clique(z')$.

\paragraph{Runtime.}
We represent the lists $L(\node)$  by linked lists, and 
form list $L(\node)$ in Step~1 by concatenation in constant time per child.   
The lists $L(c)$ of the children are no longer available afterwards, but are also no longer needed.
The transfer of the list $L(\node)$ to $\clique(z)$ in Step~3 can also
be carried out in constant time
since  $L(\node)$ is simultaneously emptied.
\iflong \else \looseness-1 \fi

\begin{theorem}
Let $G$ be an $n$-vertex RDV graph with a given RDV representation of size $O(n)$.
Then the algorithm computes a maximum independent set $I$ of $G$ together with a
clique cover of size $|I|$ 
in $O(n)$ time.
\end{theorem}
\begin{proof}
 We only add a vertex $z$ to $\I$ if $\state(z)=\available$, and by
 Invariant (I1) this means that $z$ is not adjacent to any
 vertex in~$I$. Hence $I$ is an independent set.

 By property (I4), all vertices are $\blocked$ at the end of the algorithm, and Invariant~(I5) in
 connection with Invariant~(I3) implies that 
 the family
 $\{\,\clique(z): z\in I\,\}$ forms a \emph{clique cover},
i.e., it consists
of vertex-disjoint cliques that together cover all vertices of~$G$.
By construction,
the clique cover has size~$|\I|$.

A clique cover is a certificate of optimality for~$I$:
No independent set can contain more than one 
vertex in a clique,
so no independent set can be larger than any clique cover.    So $\I$
is a maximum independent set.
\iflong
\Guenter{Noch ein schöner Extrasatz:}The existence of a clique cover of size $|I|$ is in
line with the fact that RDV graphs are perfect graphs.
\fi

The runtime is $O(n)$\iflong{} since we spend \else: we spend
\fi constant time per node and per child and $|\calT|\in O(n)$.%
\end{proof} 

A streamlined version of the algorithm (without the
clique cover) is shown in Algorithm~\ref{alg:MIS2} in
Appendix~\ref{sec:MIS-abstract}.

\paragraph{Example.}
Consider the running example of Figure~\ref{fig:running}, with the
nodes processed in the order
$b(z_7)$, $b(z_1)$, $b(z_4)$, $t(z_2)$, $t(z_1)$, $b(z_5)$, $b(z_3)$,
$t(z_3)$, $b(z_6)$, $t(z_5)$, $t(z_4)$, $t(z_7)$, $t(z_6)$.
The first top node to be processed is $t(z_2)$, where
$L(t(z_2))=\{z_1,z_2,z_4\}$; so $z_2$ is added to~$I$, all three
vertices become \blocked, and $\clique(z_2)=\{z_1,z_2,z_4\}$.
At $t(z_1)$, vertex $z_1$ is \blocked, and nothing happens.
At $t(z_3)$, we have $L(t(z_3))=\{z_3,z_5\}$, so $z_3$ is added to~$I$
with $\clique(z_3)=\{z_3,z_5\}$.
At $t(z_5)$ and $t(z_4)$, the vertices $z_5$ and $z_4$ are \blocked.
At $t(z_7)$, we have $L(t(z_7))=\{z_6,z_7\}$, so $z_7$ is added to~$I$
with $\clique(z_7)=\{z_6,z_7\}$; at the root $t(z_6)$, vertex $z_6$ is already \blocked.
The result is $I=\{z_2,z_3,z_7\}$, together with the clique cover
$\{z_1,z_2,z_4\}$, $\{z_3,z_5\}$, $\{z_6,z_7\}$ of the same size.

\section{Minimum Dominating Set \texorpdfstring{$\D$}{D}
}
\label{sec:mds3}

Booth and Johnson \cite{BoothJ82} appear to have been the first to
study {\sc DominatingSet} in RDV graphs. They gave an algorithm
(repeated in Algorithm~\ref{alg:DS_BJ})
that is a bottom-up traversal of the tree $\calT$ of the representation.
They argue that it finds a minimum dominating
set, and that it can be implemented with linear (meaning $O(m+n)$) runtime.
Gokhale restated this algorithm as a greedy algorithm based
on a vertex ordering, which then works whenever this ordering
is a so-called strong perfect elimination order \cite{PrashantThesis}.
With this reframing, the algorithm can then be implemented with $O(n\log n)$
runtime in RDV graphs \cite{BG-SOFSEM26}.   
Cheng, Kaminski and Zaks \cite{ChengKZ98} claimed 
an $O(n)$ time algorithm for {\sc DominatingSet} in a subclass of
path graphs unrelated to RDV graphs.   For this, they used as a
subroutine an $O(n)$ time algorithm for {\sc DominatingSet} in interval
graphs.  Unfortunately, none of the references that they give seems to contain such
an algorithm; but $O(n)$ runtime for
interval graphs can also be achieved by modifying the techniques for RDV graphs slightly
\cite{BG-SOFSEM26}.


\begin{algorithm}[ht]
  \caption{Greedy algorithm for {\sc DominatingSet} from \cite{BoothJ82}}
\label{alg:DS_BJ}
\Input{RDV graph $G$ with RDV representation $(\calT,t,b)$}
initialize $\D:=\emptyset$, mark all vertices as \texttt{undominated}\;
\For{each node $\node$ of $\calT$, in leaf-to-root order}
{\If{$t(z)=\node$ for some vertex $z$ and $z$ is $\mathtt{undominated}$ 
}{
	let $y\in \beta(\node)$ be the vertex with highest top node, i.e., minimizing $\depth(t(y))$\;
	add $y$ to $\D$ and mark all its neighbors as \texttt{dominated} 
}
}
\Output{$\D$}
\end{algorithm}

The algorithm that we give below for {\sc DominatingSet} in RDV graphs is the same in spirit as Algorithm~\ref{alg:DS_BJ},
but we pre-compute some of the information so that the
runtime is reduced from the (obvious) $O(m+n)$ to $O(n)$.
Like Algorithm~\ref{alg:DS_BJ},
the algorithm performs a bottom-up traversal of $\calT$
and incrementally builds a dominating set.
To facilitate the correctness argument,
we distribute this dominating set among the nodes, 
and associate with every node a ``local'' dominating 
set $D(\node)$; eventually these will all get merged into one set that we return.
Additionally the algorithm
keeps track of three vertices $\highest,\dom,\undom$, and three
integers $\h,\d,\u$ at each  node $\node$ of~$\calT$.   The intended meaning is that the following hold
at the entry to Step~3 of processing~$\node$:
\begin{description}
    \item[(D1)] $\highest(\node)$ is the vertex in bag $\beta(\node)$ whose top node is {h}ighest, i.e., minimizes the depth of the top node.
    \item[(D2)] $\dom(\node)$ is a {d}ominating vertex in the bag (i.e., $\dom(\node)\in \beta(\node)\cap \D(\node)$), and among those, the one whose top node is highest. 
    \item[(D3)] $\undom(\node)$ is an {u}ndominated vertex in the bag (i.e., $\undom(\node)\in \beta(\node)\setminus N[\D(\node)]$), and among those, the one whose top node is lowest. 

      
\item[(D4)]
$h(\node)$, $d(\node)$, and $u(\node)$ are the depths of the top nodes of $\highest(\node),\dom(\node)$ and $\undom(\node)$, respectively.
\end{description}
If any of these vertices does not exist, we set
$\highest(\node)$, $\dom(\node)$, or $\undom(\node)$ to $\NIL$,
and accordingly set
$\h(\node)=\infty$, $\d(\node)=\infty$, or $\u(\node)=-\infty$.%
\footnote{We write $\pm\infty$ for convenience; it suffices to use values outside the
  range of depths of nodes of~$\calT$.}


\paragraph{Processing node~$\node$.}  We assume that all
strict descendants of $\node$  have already been processed.
\begin{enumerate}
\item (Accumulate the data from the children.)  
If $\node$ is a leaf: 
\begin{itemize}
\item Initialize $D(\node)$ as an empty set.
\item Set  $\highest(\node):=\dom(\node):=\undom(\node):={\NIL}$, $\h(\node):=\d(\node):=\infty$, and $\u(\node):=-\infty$.
\end{itemize}
Otherwise ($\node$ has children):
  \begin{itemize}
  \item[h.]  Let $c^*:=\arg\!\min\{\h(c): \text{$c$ is a child of
      $\node$}\}$.   Set $\highest(\node):=\highest(c^*)$ and  $\h(\node):=\h(c^*)$. 
  \item[D.] Form $D(\node)$ by concatenating the lists $\D(c)$ of the children $c$ of~$a$.
  \item[d.]  Let $c^*:=\arg\!\min\{\d(c): \text{$c$ is a child of
      $\node$}\}$.   Set $\dom(\node):=\dom(c^*)$ and $\d(\node):=\d(c^*)$.
  \item[u.]
    If $\d(\node)=\infty$,
    \\\null\qquad
  set $c^*:=\arg\!\max\{\u(c): \text{$c$ is a child of $\node$}\}$,
  $\undom(\node):=\undom(c^*)$, 
  $\u(\node):=\u(c^*)$.\\
 Otherwise set $\undom(\node):=\NIL,
 \u(\node):=-\infty$.  
  \end{itemize}
\item (Update if a path starts.)
If $\node$ is a bottom node, say $\node=b(z)$:   
  \begin{itemize}
	\item[h.] If $\depth(t(z))<\h(\node)$ update $\highest(\node):=z$, $h(\node):=\depth(t(z))$.
    \item[D.] /\!/ No update is necessary for $D(\node)$. 
	\item[d.] /\!/ No update is necessary for $\dom(\node)$ and $\d(\node)$.
	\item[u.] If $\d(\node)=\infty$ and $\depth(t(z))>\u(\node)$ update $\undom(\node):=z$,  $\u(\node):=\depth(t(z))$.
  \end{itemize}
\item (Add to output if needed.)   If $\u(\node)=\depth(\node)$:

(In this case,
  $z:=\undom(\node)$ has $t(z)=\node$ and is not dominated by $D(\node)$.)
  \begin{itemize}
	\item[h.] /\!/ No update is necessary for $\highest(\node)$ and $\h(\node)$.
	\item[D.] Add $\highest(\node)$ to $\D(\node)$.
	\item[d.] Update $\dom(\node):=\highest(\node), \d(\node):=\h(\node)$.
	\item[u.] Update $\undom(\node):=\NIL, \u(\node):=-\infty$.
  \end{itemize}
\item (Update if a path ends.)
If $\node$ is a top node:
  \begin{itemize}
	\item[h.] If $\h(\node)=\depth(\node)$ then set $\highest(\node):=\NIL, \h(\node):=\infty$.
    \item[D.] /\!/ No update is necessary for $D(\node)$.
	\item[d.] If $\d(\node)=\depth(\node)$ then set $\dom(\node):=\NIL,\d(\node):=\infty$.
	\item[u.] /\!/ No update is necessary for $\undom(\node)$ and $\u(\node)$.
  \end{itemize}
\end{enumerate}
After processing  the root node~$\node$, we return $\D(\node)$ as the minimum dominating set.

\medskip
Clearly the time to process a node of $\calT$ is proportional to its number of children, 
so the runtime of this algorithm is $O(|\calT|+n)=O(n)$.   

\paragraph{Correctness.}
The proof that the invariants are maintained is straightforward.
Directly from the algorithm it is clear that Invariant~(D4) holds, i.e., 
$h(\node)$, $d(\node)$, and $u(\node)$ are the depths of the top nodes of $\highest(\node),\dom(\node)$ and $\undom(\node)$, respectively.
To prove the other invariants, we use again the set $\B(\node)$
defined in Section~\ref{sec:structure}, which
equals $\beta(\node)$ during Step~3
(Lemma~\ref{lem:beta}).

Invariant (D1):   It suffices to show that $\highest(\node)=\arg\!\min\{\depth(t(z)): z\in \B(\node)\}$ at all times.
We add to $\B(\node)$ during Steps~1 and~2, by either combining the sets of the children or adding
a vertex whose bottom node is~$\node$.   During both steps we update $\highest(\node)$, if needed, to take into
account the vertices that were newly added to $\B(\node)$.  Step~3 does not change the bag.
During Step~4, if we had \emph{any} vertex $y$ left in $\B(\node)$ after removing the one whose path ends, then $t(y)$ would be above $\node$ since
top nodes are distinct.  So $\h(\node)=\depth(\node)$ implies that $\B(\node)$ is now empty, justifying that
we set $\highest(\node)=\NIL$.%
\iflong\footnote{In particular this implies that no path in the subtree rooted at $\node$ intersects a
path in the rest of~$\calT$.    Thus in a connected graph this update step is used only when we have
processed the entire graph, and could be omitted.}\fi

Invariant (D2):   It suffices to show that $\dom(\node)=\arg\!\min\{\depth(t(z)): z\in \B(\node)\cap \D(\node)\}$ at all times.
We add to $\B(\node)$ during Steps~1 and~2.   During Step~1 we take all added vertices into account.   During
Step~2, the newly added vertex $z$ is not in $\D(\node)$, so no update is needed.   During Step~3, since some vertex in $\beta(\node)$
was undominated, none of the vertices in $\beta(\node)$ were in $\D(\node)$ previously, so the newly added $\highest(\node)$ is
the \emph{only} candidate for $\dom(\node)$.   During Step~4, as for (D1) one argues that $\d(\node)=\depth(\node)$ implies
that $\B(\node)\cap \D(\node)$ is now empty, justifying that we set $\dom(\node)=\NIL$.

Invariant (D3):  
During Steps~1 and~2, if $\d(\node)<\infty$ then we must have at least one dominating vertex in $\B(\node)$.
This dominates \emph{all} vertices in $\B(\node)$ and is in $\D(\node)$, so $\undom(\node)$ is undefined.
The same argument applies to Step~3, where $\d(\node)<\infty$ holds always.   We claim that during Step~4, no update is needed to $\undom(\node)$.
To see this, assume that $\node=t(z)$ is a top node
 and $z=\undom(\node)$,
so that the removal of $z$ from $B(\node)$ would necessitate finding a new
vertex $\undom(\node)$.
 This means that $\u(\node)=\depth(t(z))=\depth(\node)$.
  But then we would have acted in Step~3, ensuring that $z$
 becomes dominated and updating  $\undom(\node)$ accordingly.

Finally, we prove that our algorithm computes the same
dominating set as Algorithm~\ref{alg:DS_BJ}.
More specifically, we claim that we add
a vertex 
to $D(a)$ in Step~3 if
and only if  Algorithm~\ref{alg:DS_BJ} adds
a vertex 
to~$D$.
(In this case, both algorithms obviously add the same vertex, $\highest(\node)$.)


We shall show that the following four statements are equivalent:
\begin{enumerate}
\item Our algorithm adds a vertex to $D(\node)$ in Step~3.
\item The vertex $z$ with $t(z)=\node$
  is not dominated by $D(\node)$.
\item The vertex $z$ with $t(z)=\node$
  is not dominated by the global set~$D$.
\item 
  Algorithm~\ref{alg:DS_BJ} adds a vertex to $D$ when processing~$\node$.
\end{enumerate}
The equivalence between 3 and 4 follows directly from the
description of
Algorithm~\ref{alg:DS_BJ}.
However,
the equivalence between 1 and 2 is not immediately
apparent from the test in Step~3. It can be seen as follows:

A vertex is added to $\D(\node)$ when we are processing $\node$ and $\u(\node)=\depth(\node)$.   Since
$\u(\node)>-\infty$, the vertex $z:=\undom(\node)\neq \NIL$ exists, and since
$\u(\node)=\depth(\node)$ we have $t(z)=\node$ by~(D4).
By (D3) we conclude that $z$ is undominated by $\D(\node)$.
Conversely, if
$z$ with $t(z)=\node$ is undominated by $\D(\node)$, it follows
immediately from (D3) and (D4) that $\undom(\node)=z$ and
$\u(\node)=\depth(t(z))=\depth(\node)$,
and hence the test in Step~3 will be positive.

For this, we need another observation:
\begin{description}
    \item[(D5)] For every vertex $z\in D(\node)$, the path $P(z)$ contains nodes in the subtree rooted at~$\node$.
\end{description}
Inductively, this holds after concatenating
 $D(\node)$ from the sets of the children in Step~1.
 The only vertex added to $D(\node)$ is $\highest(\node)$ in Step~3, which is in $\beta(\node)$ by (D1).
 Its path therefore contains~$\node$.

To prove the equivalence between 2 and 3, we assume inductively
that $\bigcup_{c} \D(c)$ of our algorithm equals $\D$ for
Algorithm~\ref{alg:DS_BJ}. 
What remains to prove is that
if $z$ is undominated by $\D(\node)$,
it is also undominated by the global set~$\D$.
If $z$ were dominated by, say, $y\in D$, then $y\in D(b)$ for some
node $b$ that is neither an ancestor nor a descendant of~$a$. 
By~(D5), $P(y)$ contains nodes in the subtree rooted at~$b$.
But no such vertex $y$ can be adjacent  to $z$, where $\node=t(z)$,
due to Observation~\ref{obs:inside}. This contradiction proves our claim.

So the two algorithms compute the same dominating set, and since
Algorithm~\ref{alg:DS_BJ} finds a minimum dominating set \cite{BoothJ82}, so does ours.


\begin{theorem}
    Let $G$ be an $n$-vertex RDV graph with a given RDV representation of size $O(n)$.   Then we can find a minimum dominating set of $G$ in $O(n)$ time.
    \qed
\end{theorem}

We note that our algorithm can be slightly simplified:
the vertex references $\dom$ and $\undom$ are not 
needed as long as we store the depths of their top nodes, i.e., $d$ and~$u$.   See also Appendix~\ref{app:shorterAlgoDS}.

\paragraph{Example.}
Table~\ref{tab:DS-example} shows the run of the algorithm on the
running example of Figure~\ref{fig:running}, with the nodes processed in the same
order as in the example for {\sc IndependentSet}.
At $t(z_2)$, the bag contains the undominated vertex $z_2$ itself, so
$u(t(z_2))=5=\depth(t(z_2))$, and the vertex $\highest=z_4$, whose top
node is highest in the bag, is added to~$D$.
Similarly, $z_5$ is added at $t(z_3)$.
At $t(z_4)$, the dominating vertex $z_4$ leaves the bag, so
Step~4 resets $\dom$ to $\NIL$ and $d$ to $\infty$.
Hence, at $t(z_7)$, the vertex $z_7$ is undominated, and $z_6$ is added.
The result is $D=\{z_4,z_5,z_6\}$.
It is a minimum dominating set: the closed neighborhoods
$N[z_2]=\{z_1,z_2,z_4\}$, $N[z_3]=\{z_3,z_5\}$, and $N[z_7]=\{z_6,z_7\}$
are pairwise disjoint, so every dominating set has at least three vertices.

\begin{table}[ht]
\centering
\renewcommand{\arraystretch}{1.1}
\begin{tabular}{@{}llclllll@{}}
\hline
iteration&
node $\node$ & $\depth(\node)$ & bag $\beta(\node)$ & $h(\node)$ & $d(\node)$ & $u(\node)$ & added to $D$\\
\hline
1&$b(z_7)$ & 2 & $z_7$ & $1$ ($z_7$) & $\infty$ & $1$ ($z_7$) & \\
2&$b(z_1)$ & 6 & $z_1$ & $4$ ($z_1$) & $\infty$ & $4$ ($z_1$) & \\
3&$b(z_4)$ & 6 & $z_4$ & $2$ ($z_4$) & $\infty$ & $2$ ($z_4$) & \\
4&$t(z_2)$ & 5 & $z_1,z_2,z_4$ & $2$ ($z_4$) & $\infty$ & $5$ ($z_2$) & $z_4$\\
5&$t(z_1)$ & 4 & $z_1,z_4$ & $2$ ($z_4$) & $2$ ($z_4$) & $-\infty$ & \\
6&$b(z_5)$ & 6 & $z_5$ & $3$ ($z_5$) & $\infty$ & $3$ ($z_5$) & \\
7&$b(z_3)$ & 5 & $z_3,z_5$ & $3$ ($z_5$) & $\infty$ & $4$ ($z_3$) & \\
8&$t(z_3)$ & 4 & $z_3,z_5$ & $3$ ($z_5$) & $\infty$ & $4$ ($z_3$) & $z_5$\\
9&$b(z_6)$ & 4 & $z_6$ & $0$ ($z_6$) & $\infty$ & $0$ ($z_6$) & \\
10&$t(z_5)$ & 3 & $z_4,z_5,z_6$ & $0$ ($z_6$) & $2$ ($z_4$) & $-\infty$ & \\
11&$t(z_4)$ & 2 & $z_4,z_6$ & $0$ ($z_6$) & $2$ ($z_4$) & $-\infty$ & \\
12&$t(z_7)$ & 1 & $z_6,z_7$ & $0$ ($z_6$) & $\infty$ & $1$ ($z_7$) & $z_6$\\
13&$t(z_6)$ & 0 & $z_6$ & $0$ ($z_6$) & $0$ ($z_6$) & $-\infty$ & \\
\hline
\end{tabular}
\caption{The algorithm for {\sc DominatingSet} on the running example of
  Figure~\ref{fig:running}. The values $h$, $d$, and $u$ are those at the
  beginning of Step~3; the corresponding vertices $\highest$, $\dom$, and
  $\undom$ are given in parentheses.}
\label{tab:DS-example}
\end{table}

\section{Matching}
\label{sec:match}

There is a well-known greedy algorithm for {\sc Matching} (repeated
below as Algorithm~\ref{alg:Match_greedy1}).   
It uses the notion of a \emph{free} vertex, 
which is any vertex that
is not incident to any edge in the (current
) matching.
In an RDV graph, the greedy algorithm computes a maximum matching. (In fact, it does so in any strongly chordal graph if one uses a strong
perfect elimination order as vertex ordering \cite{DahlhausKarpinski}.)   In an interval graph, the runtime can be improved
to $O(n\log n)$ by reformulating the algorithm a bit \cite{MoitraJ89}. By replacing `binary search tree' in this reformulation with `van Emde Boas tree',
the runtime can be reduced further to $O(n\log\log n)$ \cite{RheeLiang}.  A runtime of $O(n\log\log n)$ can also be achieved (with
a different approach) for circular arc graphs  \cite{LiangRhee} and permutation graphs \cite{RheeLiang}.
For RDV graphs, the greedy algorithm can be implemented with $O(n\log n)$ runtime \cite{BG-CGT26}.   

\begin{algorithm}[ht]
\caption{Classic greedy algorithm for {\sc Matching} (see e.g.~\cite{DahlhausKarpinski})%
}
\label{alg:Match_greedy1}
\Input{RDV graph $G$ with RDV representation $(\calT,t,b)$}
let $z_1,\dots,z_n$ be any
leaf-to-root ordering
of $V(G)$\;
initialize the matching $\M:=\emptyset$\;
\For{$i=1,\dots,n$}{ \If{$z_i$ is free and has free  neighbors}{
		among all free neighbors of $z_i$, let $z_q$ be the
                one that minimizes $q$ /\!/ ($q>i$) \;
		add $z_iz_q$ to $\M$\;
	}
}
\Output{$\M$}
\end{algorithm}

In this section, we give a faster algorithm for {\sc Matching} in RDV graphs.  
At first glance, one might be tempted to do this by ``translating'' Algorithm~\ref{alg:Match_greedy1} 
into a traversal in~$\calT$:
Traverse the
nodes in leaf-to-root order, act on $z_i$ when we reach $t(z_i)$, and search for $z_q$ among
the free vertices of $\beta(t(z_i))$.    However, searching through this entire set would take too long.   Pre-computation of $z_q$ (similar to what we did with $\undom$) does not work either, because once $z_q$ is matched, we would need to update this reference (for passing to the parent), 
and for this would again have to inspect all free vertices of $\beta(t(z_i))$.%
\iflong
\footnote{This was not a problem
  in our algorithm for {\sc DominatingSet},
  because the vertex added to $D$ dominates
every vertex
in the bag and hence the updated $\undom$ is simply $\NIL$.}%
\fi{}
With a priority queue this could be done in $O(\log n)$ time per node
(details omitted), 
but we want to be faster than that.

We therefore use an entirely different approach based on a reformulation of
the greedy algorithm that delays adding a matching edge until we have seen
the second endpoint of it.
See Algorithm~\ref{alg:Match_greedy2-prime} for this reformulation,
which was given for an interval representation by Liang and Rhee
\cite{LiangRhee} and stated using a vertex ordering by Biedl and
Gokhale \cite{BG-CCCG24}.   Neither reference formally proved  that the two algorithms produce the same matching; for completeness we give such
a proof in  \autoref{sec:equiv}, together with
a self-contained optimality proof. 

\begin{algorithm}[ht]
  \caption{The delayed greedy algorithm for {\sc Matching} \cite{LiangRhee}}
\label{alg:Match_greedy2-prime}
\Input{RDV graph $G$ with RDV representation $(\calT,t,b)$}
let $z_1,\dots,z_n$ be any
leaf-to-root ordering
of $V(G)$\;
initialize the matching $\M:=\emptyset$\;
\For{$q=1,\dots,n$}{
	\If{$z_q$ has free neighbors $z_i$ with $i<q$}
{
		among such neighbors of $z_q$, let $z_i$ be the
                one that minimizes $i$ /\!/ ($i<q$) \;
		add $z_iz_q$ to the matching $\M$
	}
}
\Output{$\M$}
\end{algorithm}

The crucial operation
for carrying out Algorithm~\ref{alg:Match_greedy2-prime}
is to find,
among the free neighbors of $z_q$, 
the neighbor $z_i
$ with the smallest $i$,
i.e., with the top node $t(z_i)$ as low as possible,
or to detect that there are no such neighbors with $i<q$.
In terms of the tree $\calT$, this operation can be interpreted as
follows: We start from $b(z_q)$ and walk upwards until we hit the first
node that is a top node of a free vertex.
We denote this operation of ``searching
for the first  top node of a free vertex, starting
from a node $c$ and going upward'' by $\search(c)$.

If this walk reaches $t(z_q)$ without finding such a node, then $z_q$ has no free neighbor $z_i$ with $i<q$ since the top node of such a $z_i$ would have to be a strict descendant of $t(z_q)$ by Observation~\ref{obs:inside}, but an ancestor of $b(z_q)$ since $z_iz_q$ is an edge.
Vice versa, if $z_q$ has no free neighbor $z_i$ with $i<q$, then $\search(b(z_q))$ will stop at $t(z_q)$ itself,
because $z_q$ is free and thus $t(z_q)$ is a candidate node for stopping.   
On the other hand, if the walk stops at a top node $t(z_i)$ of some free vertex~$z_i$ with
$i<q$, then this must be the desired matching partner according
to Algorithm~\ref{alg:Match_greedy2-prime}, since the
top nodes $t(z_i)$ along the path from
$b(z_q)$ to $t(z_q)$ appear in increasing order of~$i$.
In summary, Algorithm~\ref{alg:Match_greedy2-prime} is 
equivalent to Algorithm~\ref{alg:match_new_search} given below.

\begin{algorithm}[ht]
  \caption{The delayed greedy
    algorithm for {\sc Matching} in terms of the
    \search-operation}
  \label{alg:match_new_search}
\Input{RDV graph $G$ with RDV representation $(\calT,t,b)$}
initialize the matching $\M:=\emptyset$\;
\For{each node $a\in V(\calT)$, in leaf-to-root order}{
  \If{$a=t(z)$ for some vertex $z$}{
    compute $c:=\search(b(z))$\;
    \If {$c\neq a$}{
        let $y$ be the vertex with $t(y)=c$\;
        add edge $yz$ to $\M$\;
      }
    }
}
\Output{$\M$}
\end{algorithm}

So we now need
a data structure that facilitates 
$\search(a)$.
An idea for this is illustrated in \autoref{fig:tree}.
Here the top nodes of free vertices are marked as thick
    bullets.
  Severing the edges between these nodes and their parent nodes
  (dashed edges) splits the tree into connected components.
  For any node $\node$ in a given component, 
  $\search(\node)$ yields the same result,
  namely the bulleted node at the root of the component.
  The task is thus reduced to identifying the component in which
  node~$\node$ lies.

\begin{figure}[htb]
  \centering
  \includegraphics{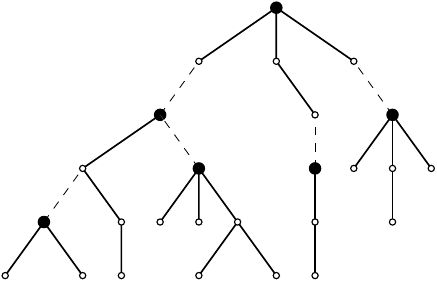}
  \caption{A tree $\calT$ 
    in which the top nodes of free vertices are marked as thick
    bullets.
}
  \label{fig:tree}
\end{figure}

  Every addition of an edge to $M$ 
means that two vertices are no longer free, so
  two bullets disappear, and
  the corresponding components are 
  merged with their parents.
This calls for a \Union-\Find data structure.
Recall that this is a data structure that maintains a partition 
of a ground set, which, in our case, is $V(\calT)$.
Each set has a \emph{representative}, and the operation $\Find(\node)$ (for any node $\node$) retrieves
the representative of the set containing~$\node$.
Initially, the partition contains
for every node
$\node$ a singleton set $\{\node\}$
whose representative is~$\node$.
The operation $\Union(a\leftarrow c)$
forms the union of the set $A$ containing $a$ and the set $C$ containing~$c$;  the new common representative is the one of set~$A$.
We will only perform this when 
we know that
$\Find(a)\ne\Find(c)$.

We will use the \Union-\Find data structure 
to perform
operation $\search$ efficiently.
All our unions will be of
the form $\Union(\parent(\node)\leftarrow\node)$,
where $\parent(\node)$ denotes the parent of~$\node$.
Recall that we assumed that the root is a top node,
so $\parent(a)$ is well-defined when $a$ is not a top node.

We will maintain the following invariants.
\begin{description}
\item[(M1)]
As long as the root of $\calT$ has not yet been processed, 
the set 
of  nodes $\node$ for which
the operation
$\Union(\parent(\node)\leftarrow\node)$
has been performed
is 
$V(\calT)\setminus\{\,t(z):\text{$z\in V(G)$, 
  $z$ is free}\,\}$.
  
\item[(M2)] As long as the root of $\calT$ has not yet been processed, 
 $\Find(c)=\search(c)$ for all nodes~$c$.
\end{description}
%

\smallskip

Our algorithm 
stores no \iflong significant helper \fi information other than 
the \Union-\Find data structure
and is just a few lines long. 
(For reference, 
Algorithm~\ref{alg:match_UF} in Appendix~\ref{app:match-UF}
gives the complete pseudocode.)

\paragraph{Initialization.} Initialize $\M:=\emptyset$.\\
Initialize 
a
\Union-\Find data structure where every node of $\calT$
forms a singleton set.\\
For every node $\node$ of $\calT$
that is not a top node,
perform 
$\Union(\parent(\node)\leftarrow\node)$.    

\paragraph{Processing node~$\node$.}
We assume that all
strict descendants of $\node$  have already been processed.
\\
If $\node=t(z)$ for some~$z$:
\\
\null\qquad    Compute $c:=\Find(b(z))$
\\\null\qquad
If $c\ne\node$:
\\\null\qquad\qquad $\bullet$ %
  Let $y$ be the vertex with $t(y)=c$, and add edge $yz$ to~$\M$.   
\\\null\qquad\qquad $\bullet$ %
\Union$(\parent(c)\leftarrow c)$.
\\\null\qquad\qquad $\bullet$ %
If $a$ is not the root of $\calT$ then perform \Union$(\parent(a)\leftarrow a)$. \\
After processing 
the root node, we return $\M$ as the maximum matching.

\paragraph{Correctness.}
Invariant (M1) is obviously satisfied: Initially, all vertices are
free,
and the initialization step performs the unions for all  nodes that
are not top nodes. Whenever a vertex becomes matched, the union operations are
performed
for the corresponding top node.
The only exception is that we do not apply the
$\Union(\parent(a)\leftarrow a)$
when
$a$ is the root of $\calT$ 
(for obvious reasons). But this happens
only in the last iteration, when the
root 
is
 processed, an exception that is provided for in Invariant~(M1).

To derive Invariant (M2) from (M1), we use the following statement:

\begin{proposition}
\label  {lem:find-correct}
Let $\calT$ be a tree rooted at node $\rho$.
Let
$F\subseteq V(\calT)$ be 
a set of nodes that contains the root $\rho$,
and let $U :=    V(\calT)\setminus F$ be its complement.
Assume we have initialized a \Union-\Find data structure with each
$a\in V(\calT)$ as a singleton set. 
After performing
$\Union(\parent(\node)\leftarrow\node)$ for all nodes $\node\in
  U
 $
  in any order,
  $\Find(c)$
  will report for any node $c$ the first node in $F$
  on the path from
  $c$ to the root~$\rho$.
\end{proposition}
\begin{proof}
See Figure~\ref{fig:tree} for an illustration.
  The operation
  $\Union(\parent(\node)\leftarrow\node)$ unites the
  set containing $\node$ with the set containing its parent.
It follows by induction on $|U|$ that the sets are the connected components
  of the subgraph of $\calT$ consisting of the edges from the nodes of
  $U$ to their parents. Also (again by induction) each component
  forms a rooted tree, whose root is the representative of the set.
  By construction the set of roots is~$F$.

  Thus $\Find(c)$ can be described as walking upwards in the tree to
  the
  root of the component, which equals the first node in $F$ on the path to
  the root of the tree.
\end{proof}

Since $\search(c)$ is \iflong precisely \else also \fi the operation
of 
reporting
the first node in $F$
  on the path from
  $c$ to the root\iflong~$\rho$\fi,
  Invariant (M2) follows as long as the root is free.
This implies that our algorithm creates the same matching as
Algorithm~\ref{alg:match_new_search}.
  Invariant~(M2) does not apply after the root \iflong $a=\rho$ \fi is
 processed,
but then the algorithm is about to terminate\iflong\ anyway\fi, and there
  will be no \iflong further \else more \fi \Find operations.
\iflong \else \looseness-1 \fi

\paragraph{Runtime.}
Every
\Union operation is between a node and its parent in the tree~$\calT$.
This is
exactly the \emph{static tree set union} scenario,
for which there is
a tailor-made data structure \cite{GabowTarjan85}:
In this situation,
  on a random-access machine with unit cost measure
  and $O(\log n)$ word length,
  for a universe of size $N$,
a sequence of $N-1$ \Union operations  and
  $k$ \Find operations can be carried out
  in $O(k + N)$ time and $O(N)$ space.
  \iflong
    (The algorithm precomputes the answers
    to \Find operations on small sets, and then uses
    table look-up on
    small sets combined with a conventional
    \Union-\Find algorithm on a universe of size
    $o(N)$.)\Guenter{candidate for shortening}
  \fi
We perform at most two \Find operations for every node.   
With this, the
\Union-\Find operations take $O(|\calT|)$ time in total,
and the rest of the algorithm can be implemented with runtime $O(|\calT|+n)$ as well.

Since the greedy algorithm is known to find a maximum matching in an RDV graph \cite{DahlhausKarpinski},
we have our result for {\sc Matching};
see Appendix~\ref{sec:optimality} for a 
self-contained proof.

\begin{theorem}
    Let $G$ be an $n$-vertex RDV graph with a given RDV representation of size $O(n)$.   Then we can find a maximum matching of $G$ in $O(n)$ time.
    \qed
\end{theorem}

\paragraph{Example.}
On the running example of Figure~\ref{fig:running}, the top nodes are
processed in the order
$t(z_2)$, $t(z_1)$, $t(z_3)$, $t(z_5)$, $t(z_4)$, $t(z_7)$, $t(z_6)$.
At $t(z_2)$, we have $\Find(b(z_2))=t(z_2)$, since $z_2$ is free and
its path consists of a single node; so nothing happens.
At $t(z_1)$, $\Find(b(z_1))=t(z_2)$, so the edge $z_2z_1$ is added to $M$,
and $t(z_2)$ and $t(z_1)$ are united with their parents.
At $t(z_3)$, nothing happens, and at $t(z_5)$, $\Find(b(z_5))=t(z_3)$,
so the edge $z_3z_5$ is added.
At $t(z_4)$ and $t(z_7)$, the search returns the node itself:
the neighbors $z_1$, $z_2$, $z_5$ of $z_4$ whose top nodes lie below
$t(z_4)$ are already matched, and $z_7$ has no neighbor whose top node is
lower than $t(z_7)$.
Figure~\ref{fig:match-example} shows the components of the
\Union-\Find data structure at this point, before the root $t(z_6)$ is
processed.
Finally, $\Find(b(z_6))=t(z_4)$, and the edge $z_4z_6$ is added.
The resulting matching $\{z_2z_1,z_3z_5,z_4z_6\}$ is maximum, since $G$
has only seven vertices.

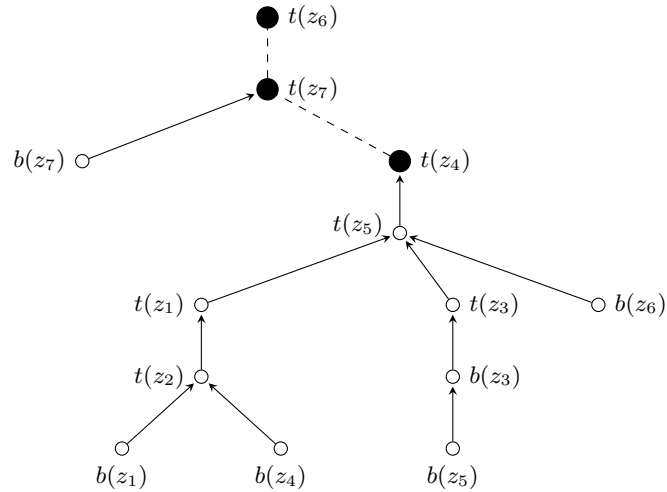
\begin{figure}[ht]
\centering
\begin{tikzpicture}[x=1.75cm,y=0.95cm,
        nd/.style = {draw,circle,fill=white,inner sep=0pt,minimum size=5pt},
        fr/.style = {draw,circle,fill=black,inner sep=0pt,minimum size=8pt},
        tree/.style = {->,>=stealth,thin,shorten >=1pt},
        cut/.style = {thin,dashed}]
\node[fr,label=right:{\small$t(z_6)$}] (r)  at (3,0) {};
\node[fr,label=right:{\small$t(z_7)$}] (s)  at (3,-1) {};
\node[nd,label=left:{\small$b(z_7)$}] (l7) at (1.6,-2) {};
\node[fr,label=right:{\small$t(z_4)$}] (p)  at (4,-2) {};
\node[nd,label=left:{\raise 1,9ex\hbox{\small$t(z_5)$}}] (q)  at (4,-3) {};
\node[nd,label=left:{\small$t(z_1)$}] (x)  at (2.5,-4) {};
\node[nd,label=left:{\small$t(z_2)$}] (y)  at (2.5,-5) {};
\node[nd,label=below:{\small$b(z_1)$}] (b1) at (1.9,-6) {};
\node[nd,label=below:{\small$b(z_4)$}] (b4) at (3.1,-6) {};
\node[nd,label=right:{\small$t(z_3)$}] (g)  at (4.4,-4) {};
\node[nd,label=right:{\small$b(z_3)$}] (b3) at (4.4,-5) {};
\node[nd,label=below:{\small$b(z_5)$}] (b5) at (4.4,-6) {};
\node[nd,label=right:{\small$b(z_6)$}] (b6) at (5.5,-4) {};
\foreach \c/\p in {l7/s,q/p,x/q,g/q,b6/q,y/x,b1/y,b4/y,b3/g,b5/b3}
  \draw[tree] (\c) -- (\p);
\foreach \c/\p in {s/r,p/s}
  \draw[cut] (\c) -- (\p);
\end{tikzpicture}
\caption{The \Union-\Find components for the running example of
  Figure~\ref{fig:running} just before the root $t(z_6)$ is processed by
  the algorithm for {\sc Matching}.
  The free vertices are $z_4$, $z_6$, and~$z_7$; their top nodes are
  drawn as bullets, and the edges from these nodes to their parents
  (dashed) have not been united.
  $\Find(b(z_6))$ returns $t(z_4)$, the root of the component containing $b(z_6)$.}
\label{fig:match-example}
\end{figure}

\section{Outlook}
\label{sec:outlook}


An obvious question is whether we 
really
need the \Union-\Find
data structure for {\sc Matching}.  (This is the most complicated part
of our algorithm, especially if we want runtime $O(n)$ rather than
$O(n\alpha(n))$.)   Can we achieve $O(n)$ runtime 
\iflong 
while storing \else  with \fi only
elementary data at each node?   
\iflong \else \looseness-1 \fi

Another open problem concerns transferring the results to other
graph classes that have intersection representations with restricted 
types of objects.
As mentioned, $O(n\log\log n)$-time algorithms are
known for {\sc Matching} in circular arc graphs \cite{LiangRhee} and 
permutation graphs \cite{RheeLiang}.
These graph classes are unrelated
to RDV graphs, and they are not even chordal.   Can 
techniques like ours be
used to reduce the runtime in these graph classes to $O(n)$ as well,
or at least to avoid the use of van Emde Boas trees?

The AI assistant Claude (Anthropic) was used during the preparation of
this paper for proofreading (consistency of notation and terminology,
typos, wording, and style), for suggestions on the organization of the
presentation, and for drafting the worked examples, including
Figures~\ref{fig:running} and~\ref{fig:match-example} and
Table~\ref{tab:DS-example}; the values in the examples were checked by
running implementations of the algorithms.
The algorithms, proofs, and results are the work of the authors, who
have checked all suggested changes and take full responsibility for
the content of the paper.


\bibliographystyle{plainurl}
\bibliography{RDV}

\appendix

\section{Streamlined Pseudocode for our Three Algorithms}
\label{sec:streamlined}


\subsection{Streamlined Algorithm for {\sc IndependentSet}}
\label{sec:MIS-abstract}

\begin{algorithm}[ht]
\caption{Streamlined algorithm for {\sc IndependentSet}}
\label{alg:MIS2}
\Input{RDV graph $G$ with RDV representation $(\calT,t,b)$}
initialize $\I:=\emptyset$, set $\state(z):=\available$ for all $z\in V(G)$, $\LL(\node):=\emptyset$ for all $\node\in V(\calT)$\;
\For{$\node\in V(\calT)$, in leaf-to-root order}{
\lFor{all children $c$ of $\node$}{append $\LL(c)$ to $\LL(\node)$ (this empties $\LL(c)$)}
\lIf{$\node=b(z)$ for some $z\in V(G)$}{add $z$ to $\LL(\node)$}
\If{$\node=t(z)$ for some $z\in V(G)$ and $\state(z)=\available$} {
	add $z$ to $\I$\; 
	\lFor{all $y\in \LL(\node)$}{delete $y$ from $\LL(\node)$ and set $\state(y):=\blocked$}
}
}
\Output{$\I$}
\end{algorithm}

\subsubsection{A Numeric Abstraction for {\sc IndependentSet}}



We can reduce the algorithm even more: 
Instead of storing a list $\LL(\node)$ we store only the lowest top node
of a vertex in $\LL(\node)$:
\begin{align*}
  \ell(\node) :={}&
  \max\ \{\, \depth(t(z)) : z \in \LL(\node)\,\}
  \\={}&
  \max\ \{\, \depth(t(z)) : z \in \beta(\node)
  \text{ and $z$ is \available}\,\}
\end{align*}
For any empty list, we take the
maximum to be the sentinel value $-1$ (which has the same effect as the
 logically sounder value $-\infty$).
The value $\ell(\node)$ allows us to find out when a vertex should be
added to~$I$.
The reason is that all vertices $z\in L(\node)\subseteq \beta(\node)$ have
$\depth(t(z))\le \depth(\node)$, by definition.
Thus, if $\ell(\node)= \depth(\node)$, it means that there must be
some \available vertex $z\in \LL(\node)$ with
$\depth(t(z))= \depth(\node)$,
and hence $t(z)=\node$.
%
\paragraph{Processing node~$\node$.}
We assume that all
strict descendants of $\node$  have already been processed.
\begin{enumerate}
\item (Accumulate the data from the children.)  
Set $\ell(\node) := \max\ \{\, \ell(c) : c \text{ a child of }\node\,\}
$.
\item (Update if a path starts.)
If $\node$ is a bottom node, say $\node=b(z)$:
\begin{itemize}
    \item 
Set $\ell(\node) := \max \{\ell(\node), \depth(t(z))\}$.
\end{itemize}
\item (Add to solution.)   If
 $\ell(\node)=\depth(\node)$, then
  $\node$ must be a top node, say $\node=t(z)$.
  \begin{itemize} 
		\item Add $z$ to~$\I$. /\!/ $z$ was \available until now.
		\item
                  Set $\ell(\node):=-1$. /\!/ Empty the list $\LL(\node)$.
  \end{itemize}
\end{enumerate}
\iflong
Algorithm~\ref{alg:MIS3} summarizes the complete algorithm in streamlined form.

\begin{algorithm}[ht]
\caption{An even more streamlined algorithm for {\sc IndependentSet}}
\label{alg:MIS3}
\Input{RDV graph $G$ with RDV representation $(\calT,t,b)$}
initialize $\I:=\emptyset$\;
\For{$\node\in V(\calT)$, in leaf-to-root order}{
set $\ell(\node) := \max\,(\{\, \ell(c) : c \text{ a child of }\node\,\}\cup\{-1\})$\;
\lIf{$\node=b(z)$ for some $z\in V(G)$}{set $\ell(\node) := \max \{\ell(\node), \depth(t(z))\}$}
\If{$\node=t(z)$ for some $z\in V(G)$ and $\ell(\node)=\depth(\node)$} {
	add $z$ to $\I$\;
	set $\ell(\node):=-1$ /\!/ Empty the list $\LL(\node)$.
}
}
\Output{$\I$}
\end{algorithm}
\fi

\subsection{Streamlined Algorithm for {\sc DominatingSet}
and {\sc Matching}}
\label{app:match-UF}
\label{app:shorterAlgoDS}
Algorithm~\ref{alg:D3} shows a streamlined version of the algorithm for {\sc DominatingSet} from Section~\ref{sec:mds3}, 
which uses a global set~$D$.

\begin{algorithm}[ht]
\caption{Streamlined algorithm for {\sc DominatingSet}}
\label{alg:D3}
\Input{RDV graph $G$ with RDV representation $(\calT,t,b)$}
initialize $D:=\emptyset$\;
\For{$\node\in V(\calT)$, in leaf-to-root order}{
set $h(\node) := \min\,(\{\, h(c) : c \text{ a child of
}\node\,\}\cup\{\infty\})$\;
\If{$\h(\node) < \infty$}{let $c^*$ be the child of $\node$ with
  $h(c^*)=h(\node)$, and
  set $\highest(\node):=\highest(c^*)$
}
set $d(\node) := \min\,(\{\, d(c) : c \text{ a child of
}\node\,\}\cup\{\infty\})$\;
{set $u(\node) := \max\,(\{\, u(c) : c \text{ a child of
  }\node\,\}\cup\{-\infty\})$}\;
\lIf{$\d(\node) < \infty$}{set $\u(\node):=-\infty$}

\If{$\node=b(z)$ for some $z\in V(G)$}{
\lIf{$\depth(t(z))<\h(\node)$}{update $\highest(\node):=z$, $h(\node):=\depth(t(z))$}
\lIf {$\d(\node)=\infty$}{update $u(a) := \max \{ u(a),
  \depth(t(z))\}$}
}
\If{$\node=t(z)$ for some $z\in V(G)$}
{
\If{$\u(\node)=\depth(\node)$}
{add $\highest(\node)$ to $\D
  $,
set $\d(\node):=\h(\node)$,
set $\u(\node):=-\infty$\;
}
\lIf{$\h(\node)=\depth(\node)$}{set $
  \h(\node):=\infty$}
\lIf {$\d(\node)=\depth(\node)$}{set $
  \d(\node):=\infty$}
}
}
\Output{$D$}
\end{algorithm}

Algorithm~\ref{alg:match_UF} shows the algorithm
for  \textsc{Matching} from
Section~\ref{sec:match} in terms of \Union-\Find operations.

\begin{algorithm}[ht]
  \caption{
    Streamlined algorithm for {\sc Matching} in terms of 
    \Union-\Find operations}
  \label{alg:match_UF}
\Input{RDV graph $G$ with RDV representation $(\calT,t,b)$}
initialize the matching $\M:=\emptyset$\;
initialize a
\Union-\Find data structure where every node of $\calT$
forms a singleton set\;
\lFor{every node $\node$ of $\calT$
that is not a top node}{
$\Union(\parent(\node)\leftarrow\node)$}
\For{each node $a\in V(\calT)$, in leaf-to-root order}{
  \If{$a=t(z)$ for some vertex $z$}{
    compute $c:=\Find(b(z))$\;
    \If {$c\neq a$}{
        let $y$ be the vertex with $t(y)=c$\;
        add edge $yz$ to $\M$\;
 $\Union(\parent(c)\leftarrow c)$\;
\lIf {$a$ is not the root of $\calT$}{\Union$(\parent(a)\leftarrow a)$}        
      }
    }
}
\Output{$\M$}
\end{algorithm}

\section{Equivalence of Algorithms~\ref{alg:Match_greedy1}
  and~\ref{alg:Match_greedy2-prime}}
\label{sec:equiv}

In this section, we argue that the greedy algorithm for {\sc Matching} (Algorithm~\ref{alg:Match_greedy1}) produces exactly the same output as the delayed greedy algorithm (Algorithm~\ref{alg:Match_greedy2-prime}).   Both algorithms were stated with a vertex ordering specifically designed for RDV graphs, but they can in fact be applied to \emph{any} vertex ordering (as long as optimality of the result is not required), and we argue equivalence even in this situation.

\begin{lemma}
  \label{lem:2-equivalent}
For any graph $G$ (not necessarily an RDV graph) and any fixed vertex ordering $z_1,\dots,z_n$,
Algorithm~\ref{alg:Match_greedy2-prime} produces the same matching as
Algorithm~\ref{alg:Match_greedy1} when both are applied with this ordering. 
\end{lemma}

\begin{proof}
The proof is by induction on the number of vertices.
If $z_1$ is an isolated vertex, it is left
unmatched by both
algorithms, and the removal 
of $z_1$ does not change anything for either algorithm.

Otherwise, 
we first argue that $z_1$ is
 matched to the same vertex $z_k$ by both
 algorithms:
Initially, all vertices are free. 
Suppose that Algorithm~\ref{alg:Match_greedy1}
matches $z_1$ to $z_k$ because $z_k$ is the lowest-numbered neighbor
of~$z_1$.
Then, when
Algorithm~\ref{alg:Match_greedy2-prime} successively considers the
vertices $z_q=z_2,z_3,\ldots,z_{k-1}$, none of these vertices is
adjacent to~$z_1$; hence $z_1$ is still free when
$z_q=z_k$ is considered, and the algorithm will of course
choose the edge $z_1z_k$.

We then observe that the removal of $z_1$ and $z_k$ does not affect
the
decisions about the remaining vertices.
This is obvious for
Algorithm~\ref{alg:Match_greedy1}, because
$z_1$ and $z_k$ are matched in the first round, and then the
vertices  $z_1$ and $z_k$ are no longer free and out of the game.

In Algorithm~\ref{alg:Match_greedy2-prime}, none of the
vertices $z_q=z_2,z_3,\ldots,z_{k-1}$ is
adjacent to $z_1$, and hence the absence of $z_1$ is not noticed.
Neighbors $z_i$ with $i>q$ such as $z_i=z_k$ do not play a role.
Later, 
when the vertices $z_q=z_{k+1},z_{k+2},\ldots$ are considered,
 $z_1$ and $z_k$ are no longer free in the first place, and hence their removal makes no difference.

 Thus, we can remove
 $z_1$ and $z_k$
 (or only $z_1$, in case $z_1$ is isolated)
 from the input without changing
 the course of the algorithms on the remaining vertices. Thus, by induction,
the remaining matching edges are also the same.
\end{proof}

\subsection{Optimality Proof for Algorithms~\ref{alg:Match_greedy1}
  and~\ref{alg:Match_greedy2-prime}}
\label{sec:optimality}
Extending the arguments in the proof of Lemma~\ref{lem:2-equivalent},
optimality for
 Algorithms~\ref{alg:Match_greedy1}
 and~\ref{alg:Match_greedy2-prime} can be established by an easy exchange argument.
 For convenience, we include this self-contained proof here.
 In contrast to Lemma~\ref{lem:2-equivalent}, we must now of course
 insist that the ordering $z_1,\dots,z_n$ 
 places $z_i$ before $z_j$ whenever
 $t(z_i)$ is a child of $t(z_j)$.
 \therese{sell this a bit better}

As mentioned in the main text, optimality of the greedy algorithm 
is known
for a
strong perfect elimination order in a
strongly chordal graph, 
due to
Dahlhaus and Karpinski~\cite{DahlhausKarpinski}.
For RDV graphs, it is not hard to show that such a strong perfect elimination order
is obtained with
the leaf-to-root order
of top nodes that we assumed in Algorithms~\ref{alg:Match_greedy1}
and~\ref{alg:Match_greedy2-prime}.
(For the reverse level order of $\calT$, this has been shown in
\cite{BG-CGT26,
PrashantThesis}.)
However, instead of deriving optimality from
the results on
strongly chordal graphs
~\cite{DahlhausKarpinski},
our proof considers 
RDV graphs directly.

 \begin{theorem}
The greedy Algorithms~\ref{alg:Match_greedy1}
 and~\ref{alg:Match_greedy2-prime} compute a maximum matching.   
 \end{theorem}
 \begin{proof}
The argument is based on the following claim:
   \begin{claim}\label{edge-optimal}
     If the algorithm matches $z_1$ with $z_q$ then there is a maximum
     matching 
     that matches $z_1$ with~$z_q$.
   \end{claim}
   From this, optimality follows by induction on $n$, along the lines
   of
   the proof of  Lemma~\ref{lem:2-equivalent}:
   If $z_1$ is matched with $z_q$, then the algorithm, when restricting our
   attention to the remaining vertices,
   behaves as if $z_1$ and $z_q$ were never present, and thus, by
   induction, it computes a maximum
   matching
   on $G\setminus \{z_1,z_q\}$. By the claim, that matching together
   with the edge $z_1z_q$ is a maximum matching for~$G$.

   If $z_1$ is an isolated vertex, then again, the algorithm behaves
   as if $z_1$ were never present, and the argument is analogous. 
 \end{proof}

\begin{proof}
  [Proof of Claim~\ref{edge-optimal}]
  Let $w_1,w_2,\ldots,w_k$ be the neighbors of $z_1$, in the order
  in which they appear in the list $z_1,\ldots,z_n$. Thus, the
  algorithm matches $z_1$ with $w_1=z_q$.

  We apply Observation~\ref{obs:inside} to the vertices~$w_i$:
  Let $\node=t(z_1)$. Then
%
  either    $w_i\in \beta(\node)$, or
its top node $t(w_i)$
is a strict descendant of~$\node$.
However, there are no top nodes
that are strict descendants of $\node$, because the corresponding
vertices would have to come before~$z_1$.

 We conclude that all neighbors $w_i$ are in $\beta(\node)$.
 Their
top nodes lie on the path from $\node$ to the root.
This has an important consequence: When checking adjacency
with $w_i$, only the path from $\node$ to $t(w_i)$ matters, and the
path below $\node$ can be ignored. The reason is that there are no top nodes
below~$\node$.

By our assumption on the ordering, the top nodes
$t(w_1),t(w_2),\ldots,t(w_k)$ appear in this order
on the path from $\node$ to the root, and consequently,
the paths from $\node$ to $t(w_i)$ (which determine the
neighbors of $w_i$) are nested.
From this we conclude that the neighborhoods are nested as well:
$N[w_1] \subseteq N[w_2] \subseteq \dots \subseteq N[w_k]$.
Intuitively, this means that 
the best matching partner to choose greedily for $z_1$
is $w_1$,
because all other 
$w_i$ have more neighbors and thus more possibilities to be matched in the remaining graph.

Formally, we prove the claim as follows.
Let $M^*$ be a maximum matching.
Since it is a maximum matching, it
cannot leave both $z_1$ and $w_1$ unmatched.

If $M^*$ matches only one of the vertices $z_1$ and $w_1$, we can always include
the edge $z_1w_1$, at the expense of removing the
matching edge incident to $z_1$ or $w_1$,
and we get a
matching that is no smaller than before.

Finally, suppose that 
$M^*$  matches $z_1$ with some vertex $w_i$ with $i>1$, and
$w_1$ is matched with some vertex $v\in N[w_1] $.
Then we can replace these two edges by the edges $z_1w_1$ and~$w_iv$.
The edge $w_iv$ must exist, since $v\in N[w_1] \subseteq N[w_i]$.
   \end{proof}
   
\end{document}

\iflong

\section{Spontaneous Optimality Proof for Everything}
\label{sec:everything}

\Guenter{while preparing for the Mittagsseminar}

\begin{lemma}\label{nested}
  Let $z_1$ be a vertex such that no other vertex has a top node that
  is a descendant of $a=t(z_1)$.

Then all neighbors of $z_1$ have their top nodes on the path from $a$ to the root. 
Let $w_1,w_2,\ldots,w_k$ be the neighbors of $z_1$,  ordered
according to the appearance of their top nodes on this path:
[$\depth(t(w_1))> \depth(t(w_2))> \cdots >\depth(t(w_k))$.]
Then
\begin{displaymath}
  N[z_1] \subseteq
  N[w_1] \subseteq
  N[w_2] \subseteq \dots \subseteq N[w_k].
\end{displaymath}
\end{lemma}
\begin{proof}
  see above.
\end{proof}

\begin{theorem}
  \begin{enumerate}
  \item There is a maximum independent set $I$ that contains $z_1$.
  \item There is a minimum dominating set $D$ that contains $w_k$.
  \item There is a maximum matching $M$ that contains the edge $z_1w_1$.
  \end{enumerate}
\end{theorem}
\begin{proof}
  \begin{enumerate}
  \item 
  Let $I$ be a maximum independent set. $I$ must contain some vertex
  $v\in N[z_1]$ because otherwise we could add $z_1$ to it.
  If $v\ne z_1$, we replace $v$ by $z_1$. Since
  $N[z_1]\subseteq N[v]$ by Lemma~\ref{nested}, the set remains independent.
  \item 
  Let $D$ be a minimum dominating set. $D$ must contain some vertex
 $v\in N[z_1]$.
  We replace, in $D$, all vertices $v\in N[z_1]$ by $w_k$.
The resulting set $D'$ is still dominating, because 
$N[v]\subseteq N[w_k]$ for all
 $v\in N[z_1]$, by Lemma~\ref{nested}. 
\item
(see above)
\qedhere
\end{enumerate}
\end{proof}
{[Was noch fehlt: removing $z_1$ or $w_k$
  [or the edge $z_1w_1$; da haben wir es schon argumentiert] 
  respectively does not
  change
  the course of the algorithm for the remaining part of the graph.]}
\fi

\end{document}